\documentclass[10pt]{article}
\usepackage{graphicx, enumerate, amsmath, latexsym}

\newtheorem{theorem}{Theorem}[section]

\newtheorem{lemma}{Lemma}[section]
\newtheorem{corollary}{Corollary}[section]

\newenvironment {proof} {{\bf Proof.}}{\hspace*{\fill}$\Box$\par\vspace{4mm}}

\def\r{\rightarrow}

\title{
 Analyzing  the Accuracy of the Fitch Method for
Reconstructing Ancestral States on Ultrametric Phylogenies }
\author{Louxin Zhang\thanks{Department of Mathematics,  National University of
Singapore (NUS), Singapore 117543. This work was supported by
ARF(R146-000-109-112).
 Email: matzlx@nus.edu.sg}, Jian Shen\thanks{Department of Mathematics, Texas State
University, San Marcos, TX 78666, USA. This work was partially
supported by
 NSF (CNS 0835834) and Texas Higher Education Coordinating Board (ARP
003615-0039-200). E-mail: js48@txstate.edu}, Jialiang
Yang\thanks{MPI-CAS Institute of Computational Biology, CAS at
Shanghai. Email: yangjialiang@picb.ac.cn}, Guoliang Li\thanks{
Department of Computer Science, NUS. Email: ligl@gis.a-star.edu.sg}
}

\begin{document}
\maketitle

\begin{abstract}
Recurrence formulas are presented for studying the accuracy of the Fitch method
for reconstructing the ancestral states in a given phylogenetic tree.
As their applications, we analyze the convergence of the accuracy of
reconstructing the root state  in a complete binary tree of $2^n$ as $n$ goes to infinity
 and also give a lower bound on  the accuracy of reconstructing
the root state in an ultrametric tree.

\end{abstract}



\section{Introduction}

Ancestral sequence reconstruction incorporates sequences from modern
living things  into evolutionary models to estimate the
corresponding sequence of an ancestor that died millions of years
ago. This approach to understanding proteins  was first suggested by
Zukerkandl and Pauling  in their seminal work
\cite{Pauling_Acta_1963} in 1963. With the rapid accumulation of
biomolecular sequence data and advances in computational biology,
 it has become an important approach to studying the
origin and evolution
  of genes, proteins and even whole genomes  (see for example
 \cite{Liberles_book_2007} and \cite{Thoronton_NRG_2004})

The Fitch method \cite{Fitch_SysZoo_1971} was  the first phylogenetic technique used for inferring
the ancestral states of a character when the phylogeny
that relates the ancestor to the extant species is
known \cite{Baba_MBE_84}. As a parsimony method, it estimates
the ancestral state by minimizing the total number of hypothetical
substitutions in  all branches that are used to explain the evolution
of the character states. It is efficient and
accurate for  sequences  that are reasonably similar to each other.
However, the accuracy of the Fitch method for
reconstructing ancestral states has yet to be well studied
\cite{Li_SysBiol_2008,Maddison_SysBiol_1995,Salisbury_SysBiol_2001,Zhang_JME_1997}.

In this work,  we present a set of recurrence formulas for analyzing
the reconstruction accuracy of the Fitch method (in
Theorem~\ref{mainThm}). These formulas are derived from a work of
Maddison \cite{Maddison_SysBiol_1995} (also see
\cite{Steel_Thesis_89}).
 They are simple and useful  as
demonstrated in solving two theoretical problems that  arise from studying
the reconstruction accuracy of the Fitch method.

 The first
problem is to analyze the convergence of the accuracy
of the Fitch method for reconstructing the root state in
 a complete phylogenetic tree in the equal-length
branch and two-state Jukes-Cantor model (see Section 2 for details). Let $p$
denote the conservation rate in each branch.
 In \cite{Steel_Thesis_89}, Steel showed  that, 
when the Fitch method is applied, the accuracy of reconstructing the root state
from all leave states in the complete binary tree of $2^n$ leaves converges
 as $n$ goes to infinity to $\frac{1}{2}$ if $\frac{1}{8} \leq
p\leq \frac{7}{8}$  and $  \frac{1}{2} +
\frac{1}{2}\frac{\sqrt{(8p-7)(4p-3)}}{(2p-1)^2}$  if
$\frac{7}{8}\leq p\leq 1$. This result was proved under the assumption that
suitable limits exists.  However, the existence of these limits
is not trivial.  In this paper, we fill the gap left in \cite{Steel_Thesis_89}
by proving that these limits exist.  In addition, we also show that  the
reconstruction accuracy diverges when $p\leq \frac{1}{8}$.

Complete phylogenetic trees in which all branches have equal length are
 special ultrametric trees. In an  ultrametric tree, each branch has
its own branch length $l(e)$, with conservation rate
$p(e)=\frac{1}{2}\left(1+e^{-l(e)}\right)$ in the two-state Jukes-Cantor model,
 but requiring that the sum of branch lengths is constant in each path from the
root to a leaf.
A counterintuitive fact is that the reconstruction accuracy  of the Fitch
method is not a monotonic function of the size of taxa selected for
reconstruction of the root state (even for ultrametric trees)
\cite{Li_SysBiol_2008}.  Hence, Li et al asked whether
the accuracy $\mbox{RA}_F$ of the Fitch method for reconstructing the root
state from all leaf states  is always  larger than or equal to
the conservation rate along  a  root-to-leaf path or not in an ultrametric
tree. Recently,  this problem is positively answered by Fischer and Thatte
\cite{Fischer_Man_2008}.
 In the second part of this paper, we
present a stronger lower bound on $\mbox{RA}_F$ for arbitrary ultrametric trees.
Our  bound  implies that $\mbox{RA}_F$  is not less than the accuracy of
reconstructing the root state from any three leaves in an ultrametric tree.

\section{The Fitch method and its reconstruction accuracy}


Let $C$ be a character with multiple states. Given a phylogenetic
tree $T$ of the character $C$ in which each leaf has a state, the
Fitch method estimates the root
 state from the leaf states in two steps. It first computes a subset
$S_u$ of states for each node $u$ of $T$ as follows:
  \begin{enumerate}
   \item  If $u$ is a leaf, $S_u$ contains only  the state of $u$;
   \item If $u$ is an internal node having children $v$ and $w$,
      $S_u$ is equal to $S_v \cup S_w$ if $S_v$ and $S_w$  are disjoint and
  $S_v\cap S_w$  otherwise.
 \end{enumerate}
After the subset $S_r$  for the root of $T$ is computed,
 the method selects a state as the root state  from $S_r$ randomly.
In other words, a state is selected as the root state  with probability
$\frac{1}{|S_r|}$, where $|S_r|$ denotes the number of states
contained in $S_r$.

Assume the mutation process along each branch of  the given tree
 is modeled as a stochastic process in which a state is replaced by another
with some probability. The Fitch method reconstructs correctly a root  state $s$
 from a set $D$ of  leave states only if  $s$ evolves into the leaf states
in  $D$. Hence,
the accuracy of the Fitch method for reconstructing the state of the root of $T$,
denoted by $\mbox{RA}_F(T)$,
 is defined to be the expected probability that
the Fitch method outputs a true state from  a set $D$  of leave states.
Let $\Pr_r [D|s]$ denote the probability that the root state
 $s$ evolves into the leaf states in $D$. Then,
{\small
 \begin{eqnarray}
  \mbox{RA}_{F}(T) 
=\sum_{s, D} p_r(s) \mbox{Pr} _r [ D|s] \Pr[ s \mbox{ is output from
} D],  \label{def_orginial1}
\end{eqnarray}
}
where $p_r(s)$ is the prior probability of $s$ being  the root state.

\section{Recurrence formulas  for analyzing the reconstruction accuracy}

In the rest of this paper, we assume that the character has only
 two states 0 and 1 and the root
takes these two states with equal prior probability.
By definition, the Fitch method selects 1 with probability 1 if  $\{1\}$ is
the state subset $S_r(D)$
computed from $D$ at the root in the first step. Otherwise, it selects 1 from
$S_r(D)=\{0,1\}$ with
probability $\frac{1}{2}$.
Therefore, by symmetry,
(\ref{def_orginial1}) becomes

{\small
 \begin{eqnarray}
  \mbox{RA}_{F}(T) = \sum _{D}  \mbox{Pr} _r  [ D|1]
 (\Pr[ S_r(D)=\{1\}]
 + \frac{1}{2}\Pr[ S_r(D)=\{0, 1\}]). \label{def_original}
\end{eqnarray}
}

  Let
\begin{eqnarray*}
\mbox{Pr} _X [S |s]=\sum _{D'}  \mbox{Pr}_X [ D'|s] \Pr[ S_X(D')=S ]
\end{eqnarray*}
 for a node $X$, a state $s\in \{0, 1\}$, $S=\{1\}, \{0, 1\}$,
and a set $D'$ of possible states of the leaves below $X$.
$\Pr_X [S|s]$ is the probability that the Fitch method outputs state subset $S$
at $X$ in its
first step given the true state of $X$ is $s$.
By symmetry,
{\small
\begin{eqnarray*}
  & \mbox{Pr}_X[ \{1\}|1] = \mbox{Pr}_X[ \{0\}|0], \\
  & \mbox{Pr}_X[ \{0\}|1] = \mbox{Pr}_X[ \{1\}|0],\\
  & \mbox{Pr}_X[ \{0, 1\}|1] =
   1-  \mbox{Pr}_X[ \{1\}|1] - \mbox{Pr}_X[ \{0\}|1],\\
  &  \mbox{Pr}_X[ \{0, 1\}|0] =
   1-  \mbox{Pr}_X[ \{1\}|0] - \mbox{Pr}_X[ \{0\}|0].
\end{eqnarray*}
}

For a node $X$ and a state $s=0, 1$, we further set
$$\alpha _X = \mbox{Pr} _X [\{s\}|s],~~\beta_X = \mbox{Pr} _X [\{1-s\}|s].$$
Then,
$$\mbox{Pr}_X[\{0, 1\}| s]=1-\alpha _X -\beta _X .$$
Then,  (\ref{def_original}) becomes
{\small
\begin{eqnarray}
 \mbox{RA}_{\mbox{F}}(T)&\hspace*{-1em}= &\hspace*{-1em} \mbox{Pr}_{r}[\{1\}|1] +\frac{1}{2}\mbox{Pr}_{r}[\{0, 1\}|1]\nonumber \\
&\hspace*{-1em}= &\hspace*{-1em} \frac{1}{2} + \frac{1}{2}\left(\mbox{Pr}_{r}[\{1\}|1] -\mbox{Pr}_{r}[\{0\}|1]\right) \nonumber \\
&\hspace*{-1em}= &\hspace*{-1em} \frac{1}{2} + \frac{1}{2}
(\alpha_r - \beta_r).
\label{definition}
\end{eqnarray}
}

Let $Z$ be an internal node  and have $X$ and $Y$ as its children. Furthermore, we
let the conservation probability on branches $ZX$ and $ZY$ be $p_X$ and $p_Y$,
respectively. The subset $S_Z$  computed at $Z$ is $\{1\}$ if and only if
one of $S_X$ and $S_Y$ is $\{1\}$ and the other is $\{1\}$ or $\{0, 1\}$.
Hence,
\begin{eqnarray}
  \alpha_Z &= & (p_X \alpha_X + q_X \beta_X) (p_Y \alpha_Y + q_Y \beta_Y)
\nonumber \\
          &&  +  (p_X \alpha_X + q_X \beta_X) (1-\alpha_Y -\beta_Y) \nonumber \\
          &&   +  (1-\alpha_X -\beta_X)(p_Y \alpha_Y + q_Y \beta_Y) \label{formula_alpha},
\end{eqnarray}
where $q_X=1-p_X$ and $q_Y=1-p_Y$.
Similarly
\begin{eqnarray}
   \beta_Z & = &  (q_X \alpha_X + p_X \beta_X) (q_Y \alpha_Y + p_Y \beta_Y)
 \nonumber\\
        &&    +  (q_X \alpha_X + p_X \beta_X) (1-\alpha_Y -\beta_Y)\nonumber \\
        &&     +  (1-\alpha_X -\beta_X)(q_Y \alpha_Y + p_Y \beta_Y) \label{formula_beta}
\end{eqnarray}
These two recurrence relations presented in \cite{Maddison_SysBiol_1995} lead to  an efficient dynamic
programming method for calculating $\alpha_r$ and $\beta_r$.
But, these two relations are not simple enough for the  theoretical study of
 the reconstruction accuracy. In the rest of this section,
 we shall establish two recurrence relations for the purpose of the theoretical analysis.

Let $$C_Z=1-\alpha_Z -\beta_Z$$ and
$$D_Z=\alpha_Z -\beta_Z.$$
If $Z$ is a leaf, we have that
\begin{eqnarray}
  C_Z=0,~~ D_Z=1. \label{Value_for_leaf}
\end{eqnarray}
Otherwise,  we have the following recurrence relations.

\begin{theorem}
\label{mainThm}
Let $Z$ be an internal node and have children $X$ and $Y$. Then,
{\small
\begin{eqnarray}
 C_Z & = & \frac{1}{2}\times [1 - C_X - C_Y + 3 C_X C_Y \nonumber  \\
    & & - (2p_X-1)(2p_Y-1)D_XD_Y], \label{recurrence_for_C}
\end{eqnarray}
}
and
{\small
\begin{eqnarray}
 D_{Z} & = &  \frac{1}{2}(2p_{X}-1) (1 +  C_Y) D_X \nonumber \\
  & &  + \frac{1}{2}(2p_Y-1)(1+ C_X) D_Y.  \label{recurrence_for_D}
\end{eqnarray}
}
\end{theorem}
\begin{proof} These two relations can be verified by
 using (\ref{formula_alpha}) and (\ref{formula_beta}).
The details can be found in Appendix.
\end{proof}

As the first application of this theorem, we obtain the following
fact. This result can be found in \cite{Steel_Thesis_89}. Here we
give a short proof.

\begin{corollary}
\label{C_bound}
 For any phylogenetic tree $T$ with root $r$ in which the conservation
probability is at least $\frac{1}{2}$,
$\Pr[{0, 1}|s]=C_r \leq \frac{1}{2}$ for $s=0,1$.
\end{corollary}
\begin{proof} We prove the fact by induction on $n$, the number of nodes of $T$.
For $n=0$, the fact follows from (\ref{Value_for_leaf}). Suppose
$C_r \leq \frac{1}{2}$ for any tree with less than $n$ nodes. Now, consider a
phylogenetic tree
$T$ of  $n$ nodes. Let the root $r$ of $T$  have children $X$ and $Y$.
Then, by induction, $0\leq C_X, C_Y \leq \frac{1}{2}$.
Since $p_X, p_Y\geq 1/2$, by  Formula~(\ref{recurrence_for_C}),
{\small
   \begin{eqnarray*}
   C_r& = &({1}/{2})[ {2}/{3} +
3\left(C_X-1/3\right)\left(C_Y- 1/3\right)\\
      & & -(2p_X-1)(2p_Y-1)D_XD_Y]\\
     & \leq & (1/2)\left[ 2/3 + 3 \left|C_X-1/3\right|
\times \left|C_Y-1/3\right|\right]\\
     & \leq & (1/2)\left[ 2/3 + 3 \times (1/3)^2\right]\\
     & = & 1/2.
 \end{eqnarray*}
}
Hence, the fact holds.
\end{proof}

\section{Accuracy on complete binary trees}

In this section, we study the reconstruction accuracy of the Fitch method on
the complete binary trees.
Let $T_n$ be the complete binary tree of $2^n$ leaves in which the conservation
 probability is $p$ along each branch.  Let $r$  denote the root of $T_n$ and
 $C_n(p)=C_r$ and $D_n(p) =D_r$ in $T_n$.   Since the subtree rooted at
each child of the root in $T_n$ is the complete binary tree of $2^{n-1}$
leaves, (\ref{recurrence_for_C}) and (\ref{recurrence_for_D}) imply that,
for $n\geq 1$,

{\small
\begin{eqnarray}
 2C_n (p)&\hspace*{-1em}=&\hspace*{-1em}1- 2C_{n-1}(p) + 3 C_{n-1}^2(p) -(2p-1)^2 D_{n-1}^2(p), \nonumber \\
 D_n (p)&\hspace*{-1em}=&\hspace*{-1em}(2p-1) \left(1 + C_{n-1}(p)\right)D_{n-1} (p),
 \label{recurrence_2}
\end{eqnarray}
}
where $0\leq p\leq 1$.

\begin{lemma} \label{reduction}
For any $n\ge 1$ and $0\leq p\leq 1$,
$$C_n(p) = C_n (1-p),~~|D_n(p)| = |D_n(1-p)|$$.
\end{lemma}
\begin{proof} We prove by induction on $n$. For $n=0$,
the facts follow from Formula~(\ref{Value_for_leaf}).

 Suppose now the lemma is
true for $n-1$; that is, $C_{n-1} (p) = C_{n-1} (1-p)$ and
$|D_{n-1}(p)| = |D_{n-1}(1-p)| $. Then
{\small
$$\begin{array}{ll}
& 2C_n (p)\\
  = &  1- 2C_{n-1}(p) + 3 C_{n-1}^2(p) -(2p-1)^2 D_{n-1}^2(p) \\
= & 1- 2C_{n-1}(1-p) + 3 C_{n-1}^2(1-p) \\
 & -\left(2(1-p)-1\right)^2 D_{n-1}^2(1-p)\\
= & 2C_n(1-p)
\end{array}$$
}
and
{\small
$$ \begin{array}{ll}
& |D_n (p) |\\
 = & |2p-1| \cdot \left(1 + C_{n-1}(p)\right) \cdot | D_{n-1} (p)| \\
= & |2(1-p)-1| \cdot \left(1 + C_{n-1}(1-p)\right)\cdot | D_{n-1} (1-p)|\\
 = & |D_n (1-p)|,
\end{array}$$
}
from which Lemma~\ref{reduction} follows by induction.
\end{proof}

By Lemma~\ref{reduction}, we have
$$ \lim_{n \r \infty} C_n(p) =
\lim_{n \r \infty} C_n(1-p) $$
and
$$ \lim_{n \r \infty} |D_n(p)|
= \lim_{n \r \infty} |D_n(1-p)|,$$
 if all the above limits exist.
Therefore, it suffices to assume that $1/2 \le p \le 1$. Now we
simplify our notations by dropping $p$ from two equalities in
(\ref{recurrence_2}), resulting in

\begin{eqnarray}
2C_n  = 1- 2C_{n-1}  + 3 C_{n-1}^2  -(2p-1)^2 D_{n-1}^2 , \label{recurrence_1_simple}\\
D_n   = (2p-1) (1 + C_{n-1} )D_{n-1}.  \label{recurrence_2_simple}
\end{eqnarray}


\begin{lemma}
\label{Bound_on_D}
For any $n\geq 1$,
$$0\leq C_n\leq \frac{1}{2},~~0\leq D_n \leq 1.$$
\end{lemma}
\begin{proof} Since we assume $1/2\leq p\leq1$,
the first fact is from Corollary~\ref{C_bound}. The second
fact is trivial.
\end{proof}

\begin{lemma} \label{smallC}
Let  $n \ge 1$.
If  $C_{n-1} \le \frac{1}{3}$, then $ C_n \le \frac{1}{3}$.
\end{lemma}
\begin{proof} We rewrite  Formula (\ref{recurrence_1_simple})  as
{\small
\begin{eqnarray}
\label{another_form}
2\left( \frac{1}{3}-C_n\right) +  3\left( \frac 1 3 -C_{n-1}\right)^2 -(2p-1)^2
D_{n-1}^2=0.
\end{eqnarray}
}
This implies that
{\small
$$0 \le 2\left(\frac 13 -C_{n-1}\right)
 \le  (2p-1)^2 D_{n-2}^2,$$
}
and
$$4\left(\frac 13 -C_{n-1}\right)^2 \le (2p-1)^4 D_{n-2}^4.$$
By Lemma~\ref{Bound_on_D}, we have that
{\small
$$\begin{array}{rcl}
2C_n & =&\frac  2 3 +  3\left( \frac 1 3 -C_{n-1}\right)^2 -(2p-1)^2 D_{n-1}^2 \\
& \le & \frac  2 3 +  \frac 34  (2p-1)^4 D_{n-2}^4\\
&  & -(2p-1)^2 \left[(2p-1)(1+C_{n-2}) D_{n-2}\right]^2 \\
& = &\frac 2 3 + (2p-1)^4 \left ( \frac 34 D_{n-2}^2 - ( 1+ C_{n-2})^2 \right )D_{n-2}^2 \\
& \le & \frac 2 3 + (2p-1)^4 \left ( \frac 34  - ( 1+ 0)^2 \right)D_{n-2}^2 \\
& \le & \frac 23.
\end{array}
$$
}
and hence Lemma~\ref{smallC} follows.
\end{proof}

\begin{lemma} \label{largeC}
Let  $n \ge 1$.
If $C_{n-1} \ge \frac{1}{3}$, then $C_n \le C_{n-1}$.
\end{lemma}
\begin{proof}
{\small
$$\begin{array}{lcl}
2C_n &\hspace*{-1em}=&\hspace*{-1em}1- 2C_{n-1} + 3 C_{n-1}^2 -(2p-1)^2 D_{n-1}^2\\
&\hspace*{-1em}= &\hspace*{-1em}2 C_{n-1} +(1-C_{n-1})(1-3C_{n-1}) -(2p-1)^2 D_{n-1}^2\\
& \hspace*{-1em}\le &\hspace*{-1em}2C_{n-1}.
\end{array}$$
}
\end{proof}

\begin{theorem}
\label{Thm41}
Suppose $\frac{1}{8} \le p < \frac{7}{8}$. Then
$$\lim_{n \r \infty} C_n = \frac 13,~~
 \lim_{n \r \infty} D_n =0.$$
\end{theorem}
\begin{proof}
The proof is divided into two cases.

Case 1: $C_n \ge 1/3$ for all $n$. By Lemma~\ref{largeC}, $C_n$ is a
decreasing positive sequence and thus $\lim_{n \r \infty} C_n$
exists and its value is at least $1/3$. The equality $2C_n = 1- 2C_{n-1} + 3 C_{n-1}^2 -(2p-1)^2
D_{n-1}^2$ implies that $\lim_{n \r \infty} D_n$ exists.
 Taking limits on all terms in (\ref{recurrence_2_simple}) implies that
$ \lim_{n \r \infty} D_n =0$ since $\lim_{n \r \infty} C_n \geq 1/3$.
Again, taking on all terms in (\ref{recurrence_1_simple}) gives that
$$2\lim_{n \r \infty} C_n= 1- 2\lim_{n \r \infty} C_n + 3 \left( \lim_{n \r \infty}
C_n \right )^2 -0;$$ that is, $\lim_{n \r \infty} C_n =1/3$ or $1$.
Since $C_n$ is decreasing and $C_1 = 2p(1-p) <1/2$,  $ \lim_{n \r
\infty} C_n \ne 1$. Thus $\lim_{n \r \infty} C_n =1/3$.

Case 2: $C_N < 1/3$ for some $N$. By Lemma~\ref{smallC}, $C_n \le 1/3$ for
all $n \ge N$.  Formula (\ref{recurrence_2_simple}) implies that
{\small
$$ D_n = (2p-1) (1+C_{n-1}) D_{n-1} \leq \left(\frac{4}{3}(2p-1)\right)^{n-N} D_{N-1}$$
}
for any $n\geq N$.
Since $1/2\leq p<7/8$, $\frac{4}{3}(2p-1) <1$ and hence
 $ \lim_{n \r \infty} D_n =0$.

By Formula (\ref{another_form}),
{\small
\begin{eqnarray*}
 2 \left(\frac 13 - C_n\right) = (2p-1)^2 D_{n-1}^2 - 3
\left(C_{n-1} -\frac 13 \right)^2
\end{eqnarray*}
}
and hence
{\small
\begin{eqnarray*}
 2 \left(\frac 13 - C_n\right) \le (2p-1)^2 D_{n-1}^2
\end{eqnarray*}
}
 for all $n\geq N$.
Since
$$ 0\leq 2 \left(\frac 13 - C_n\right)$$
and
$$\lim_{n \r \infty }
(2p-1)^2 D_{n-1}^2 = (2p-1)^2 \left( \lim_{n \r \infty } D_{n-1} \right )^2 =0,
$$
by the Sandwich Theorem $$\lim_{n \r \infty } 2
\left(\frac 13 - C_n\right) =0$$ and thus $\lim_{n \r \infty } C_n = 1/3.$
\end{proof}

To prove the convergence of $C_n$ and $D_n$ for $p\geq \frac{7}{8}$, we set
$$c_n = 2(1-p)/(2p-1)-C_n$$ and
$$d_n =D_n^2.$$
Then,  Formula (\ref{another_form}) implies that

$$ \begin{array}{ll}
  & 2(\frac{2(1-p)}{2p-1} - c_n)\\
  = &  \frac{2}{3} + 3 \left(\frac{1}{3}-\frac{2(1-p)}{2p-1}+c_{n-1}\right)^2
-(2p-1)^2d_{n-1}\\
  = & \frac{2}{3} + 3\left(\frac{8p-7}{3(2p-1)}+c_{n-1}\right)^2 - (2p-1)^2d_{n-1}\\
  = & \frac{2}{3}+ \frac{(8p-7)^2}{3(2p-1)^2} + \frac{2(8p-7)}{2p-1}c_{n-1}
      + 3 c_{n-1}^2 - (2p-1)^2 d_{n-1},\\
  \end{array}
$$
 or equivalently

{\small
 \begin{eqnarray}
   2c_n
  & = & (2p-1)^2d_{n-1} - \frac{2(8p-7)}{2p-1}c_{n-1} - 3c_{n-1}^2 \nonumber \\
  &   &- \frac{(8p-7)(4p-3)}{(2p-1)^2}. \label{recurrence_incase78_1}
   \end{eqnarray}
}
Formula (\ref{recurrence_2_simple}) implies that
\begin{eqnarray}
d_n &= & (2p-1)^2 \left(\frac{1}{2p-1} - c_n\right)^2 d_{n-1}\nonumber \\
     & &   = \left[ 1 - (2p-1)c_{n-1}\right]^2 d_{n-1}.  \label{recurrence_incase78_2}
\end{eqnarray}

\begin{lemma}
\label{key_lemma_incase78}
For any $k\geq 2$ and $p\geq 7/8$,

(1)  $c_{k} \geq 0$.

(2)  $d_{k+1} \leq d_k$.

(3)  $c_{k} \leq \frac{5(1-p)}{4(2p-1)}.$
\end{lemma}
\begin{proof}
 We prove it by induction on $k$. The facts is obviously true for $k=2, 3$.
Assume they hold for $k\leq n-1$. We now  prove they hold for $k=n$.

(1). By induction, $0\leq c_{n-2}, c_{n-1}\leq \frac{5(1-p)}{4(2p-1)}$. Hence,
  \begin{eqnarray}
& & [1-(2p-1)c_{n-2}]^2 - \frac{8p-7}{2p-1} - \frac{3}{2}c_{n-1} \nonumber \\
& = & \frac{6(1-p)}{2p-1} - 2(2p-1)c_{n-2}- \frac{3}{2}c_{n-1} + (2p-1)^2
c^2_{n-2} \nonumber \\
& \geq &  \frac{6(1-p)}{2p-1} - \frac{8p-1}{2} \times \frac{5(1-p)}{4(2p-1)}+0
\nonumber \\
&= & \frac{1-p}{2p-1}\times \frac{53-40p}{8} \nonumber \\
&\geq &  0. \label{tech_lemma1}
\end{eqnarray}
Setting $\Delta=\frac{(8p-7)(4p-3)}{(2p-1)^2}$, we have
$$\begin{array}{ll}
  & 2c_{n}\\
  = & (2p-1)^2 d_{n-1} - \frac{2(8p-7)}{2p-1} c_{n-1} - 3c^2_{n-1} - \Delta\\
  = & (2p-1)^2 d_{n-1} - 2c_{n-1} \left(\frac{8p-7}{2p-1}
+\frac{3}{2}c_{n-1}\right ) - \Delta.
\end{array}
$$
By using recurrence  (\ref{recurrence_incase78_1}) and
(\ref{recurrence_incase78_2}), we obtain that
$$\begin{array}{ll}
  & \hspace*{-1em}2c_{n}\\
   = &\hspace*{-1em}(2p-1)^2 \left(1-(2p-1)c_{n-2}\right)^2d_{n-2}
    - [\frac{(8p-7)}{2p-1} +  \frac{3}{2}c_{n-1}]\\
   &  \times [(2p-1)^2 d_{n-2} - \frac{2(8p-7)}{2p-1} c_{n-2} - 3c^2_{n-2}-\Delta]
-\Delta\\
  = & \hspace*{-1em}(2p-1)^2 [\left(1-(2p-1)c_{n-2}\right)^2 - \frac{8p-7}{2p-1} -
\frac{3}{2}c_{n-1}]d_{n-2}\\
 & + [\frac{8p-7}{2p-1}+ \frac{3}{2}c_{n-1}]
 [ \frac{2(8p-7)}{2p-1} c_{n-2} + 3c^2_{n-2}+\Delta] -\Delta.\\
 \end{array}
$$
Since $c_{n-1}\geq 0$, Formula (\ref{recurrence_incase78_1}) implies that
$$ (2p-1)^2d_{n-2} \geq \frac{2(8p-7)}{2p-1}c_{n-2} + 3 c_{n-2}^2 + \Delta.$$
This inequality and (\ref{tech_lemma1}) implies that
$$\begin{array}{ll}
  & 2c_{n}\\
  \geq  & [(1-(2p-1)c_{n-2})^2 - \frac{8p-7}{2p-1} - \frac{3}{2}c_{n-1}]\\
     &  \times [\frac{2(8p-7)}{2p-1}c_{n-2} + 3 c_{n-2}^2 + \Delta]\\
 & + [\frac{8p-7}{2p-1}+ \frac{3}{2}c_{n-1}]
 [ \frac{2(8p-7)}{2p-1} c_{n-2} + 3c^2_{n-2}+\Delta] -\Delta\\
    = & \frac{8(8p-7)(1-p)}{2p-1}c_{n-2} +[3+(8p-7)(4p-7)]c^2_{n-2}\\
     & + 4(2p-1)(4p-5)c^3_{n-2} +3(2p-1)^2c^4_{n-2}.
\end{array}
$$
By assumption, $c_{n-2} \leq \frac{5(1-p)}{4(2p-1)}$ and $4p-5<-1$. Replacing
$c^3_{n-2}$  with $\frac{5(1-p)}{4(2p-1)}c^{2}_{n-2}$ in the right-hand side of
the last inequality,
we have that
$$\begin{array}{ll}
  & 2c_{n}\\
    \geq &
     \frac{8(8p-7)(1-p)}{2p-1}c_{n-2} +[3+(8p-7)(4p-7)]c^2_{n-2}\\
     & + 5(1-p)(4p-5)c^2_{n-2} +3(2p-1)^2c^4_{n-2}\\
     = & \frac{8(8p-7)(1-p)}{2p-1}c_{n-2} +3(1-p)(9-4p)c^2_{n-2}\\
      & +3(2p-1)^2c^4_{n-2}\\
    \geq & 0
\end{array} $$
%

  (2) We have proved that $c_{n}\geq 0$.  Therefore, \\
$d_{n+1}=[1-(2p-1)c_n]^2d_{n}\leq d_{n}$.
\vspace{1em}

 (3)  Since $d_k$ decreases for $k\leq n$,
\begin{eqnarray}
\label{tech_lemma}
   d_n \leq d_2=D_2^2< (2p-1)^2
\end{eqnarray}
  Let $q=1-p$. Note that $p\geq \frac{7}{8}$ and $q\leq \frac{1}{8}$.
Therefore,  we have  that
  $$ \frac{1}{1-2q}\leq \frac{4}{3}$$
  and
  $$16q(1-5q) \leq 16\times  \frac{1}{10}\times
  \left(1-5\times \frac{1}{10}\right)= \frac{4}{5}.$$

Recalling that $c_{n-1}\geq 0$, by (\ref{tech_lemma}),  we have that
 $$\begin{array}{ll}
    & c_{n}\\
     =& \frac{1}{2}[(2p-1)^2d_{n-1} - \frac{2(8p-7)}{2p-1}c_{n-1} -
     3c_{n-1}^2\\
    &  - \frac{(8p-7)(4p-3)}{(2p-1)^2}]\\
    \leq &  \frac{1}{2} [(2p-1)^2 d_{n-1} -  \frac{(8p-7)(4p-3)}{(2p-1)^2}]\\
        = &  \frac{1}{2(2p-1)^2} [(2p-1)^4 d_{n-1} - (8p-7)(4p-3)]\\
    \leq & \frac{1}{2(2p-1)^2} [ (2p-1)^6 - (8p-7)(4p-3)]\\
    = & \frac{1}{2(2p-1)^2} [ (1-2q)^6 - (1-8q)(1-4q)]\\
    = & \frac{q}{(2p-1)^2} [2q(7-40q+60q^2-48q^3+16q^4)]\\
    \leq & \frac{q}{(2p-1)^2} [2q(7-40q+ 60q^2+16q^4)]\\
    \leq  &  \frac{q}{(2p-1)^2} [2q(7-40q + \frac{60}{64}+ \frac{1}{256})]\\
    \leq & \frac{q}{(2p-1)^2} [2q(8-40q)]\\
    = & \frac{q}{2p-1} \frac{16q(1-5q)}{1-2q}\\
    \leq & \frac{4q}{5(2p-1)} \frac{1}{1-2q}\\
\end{array}
$$
Since $q\leq \frac{1}{8}$ and  $\frac{1}{1-2q} \leq \frac{4}{3}$,
     $ c_{n} \leq  \frac{16q}{15(2p-1)} \leq  \frac{5q}{4(2p-1)}. $
\end{proof}

\begin{theorem}
\label{Thm42}
 Suppose $\frac{7}{8} \leq  p \leq  1$. Then
$$\lim_{n \r \infty} C_n = \frac{2(1-p)}{2p-1}$$
and
$$ \lim_{n \r \infty} D^2_n = \frac{(8p-7)(4p-3)}{(2p-1)^4}.$$
\end{theorem}
\begin{proof}
 Since  $c_n\geq 0$ for all $n$,  Formula
(\ref{recurrence_incase78_1}) implies
$$d_n\geq \frac{(8p-7)(4p-3)}{(2p-1)^4}$$ for all $n$.
Since $d_n=D^2_{n}$ is a decreasing sequence,
$\lim _{n\rightarrow \infty} d_n$ exists and
is at least $\frac{(8p-7)(4p-3)}{(2p-1)^4}$, which is larger than 0
for $p>\frac{7}{8}$.  Since $0\leq c_{n}\leq 1$,
$$0\leq 1-(2p-1)c_{n} \leq 1.$$  For $p>\frac{7}{8},$
 Formula
(\ref{recurrence_incase78_2}) implies that
 $$\lim _{n\rightarrow \infty} 1-(2p-1)c_{n} =1$$
 and so $$\lim _{n\r \infty} c_n=0.$$
Hence, $\lim_{n\r \infty } C_n = \frac{2(1-p)}{2p-1}$.

For $p=\frac{7}{8}$,
 Formulas (\ref{recurrence_incase78_1}) and
(\ref{recurrence_incase78_2}) become
$$ 2c_n + 3c^2_{n-1} = \frac{9}{16} d_{n-1}$$
and
$$ d_n = \left(1- \frac{3}{4} c_{n-1}\right)^2d_{n-1}.$$
As a decreasing sequence, $d_n$ has an non-negative limit. If $\lim
_{n\r \infty} d_n=0$, by the Sandwich theorem, $\lim_{n\r
\infty}c_n=0$ from the fact that $0\leq 2c_n \leq
\frac{9}{16}d_{n-1}$. Therefore, $$\lim_{n\r \infty}
C_{n}=\frac{2(1-p)}{2p-1}$$ and
$$\lim _{n\r \infty } D^2_{n}=\frac{(8p-7)(4p-3)}{(2p-1)^4}.$$

If $\lim _{n \r \infty} d_n >0$, then,  $$ d_n = d_{n-1} (1-
\frac{3}{4} c_{n-1})^2$$ implies that $\lim _{n\r \infty} c_n =0$
and hence $\lim _{n\r \infty} d_n=0$, a contradiction.
\end{proof}

\begin{theorem} Let $T_n$ be the complete binary tree of $2^n$ leaves in which
the conservation rate is $p$ along each branch.
In the two-state Jukes-Cantor model,

(a)  ({\bf Steel} \cite{Steel_Thesis_89})
the accuracy of the Fitch method
for reconstructing the root state in $T_n$ converges as $n$ goes to infinity
to $\frac{1}{2}+\frac{1}{2(2p-1)^2}\sqrt{(8p-7)(4p-3)}$ if
$p\in [\frac{7}{8}, 1]$ and $\frac{1}{2}$ if $p\in [\frac{1}{8}, \frac{7}{8}]$.

(b) it diverges as $n$ goes to infinity if $p\in (0, \frac{1}{8})$;
\end{theorem}
\begin{proof}
By Formula (\ref{definition}) and the definition of $D_n$,
$$ RA_F(T_n)= \frac{1}{2}+\frac{1}{2}D_{n}.$$
Hence, the fact (a) follows from Theorems~\ref{Thm41} and \ref{Thm42}.

When $0< p <\frac{1}{8}$, $D_n >0$ for even integers $n$ and $D_n<0$ for odd
integers $n$. By  Lemma~\ref{reduction} and Theorem~\ref{Thm42},
$|D_n|$ converges to  a positive number. Hence $D_n$ and $RA_F(T_n)$ diverge.
\end{proof}

\section{The reconstruction  accuracy on ultrametric trees}

We now consider the accuracy of reconstructing the root state in
ultrametric phylogenies. In  an ultrametric phylogeny $T$,  a  branch
$xy$ has a length $t_{xy}$,  but  all the leaves have the same
distance from the root.
 Under the two-state Jukes-Cantor model, the conservation
probability $p_{xy}$ along  a branch $xy$ of length $t_{xy}$ is
 $$ p_{xy}=\frac{1}{2}(1+e^{-2\lambda t_{xy}}),$$
 where $\lambda$ is a constant, representing the substitution rate
 in $T$. 
 For an internal node $u$ of $T$,
the distance between it and any of  its leaf descendants is defined
as its depth, denoted by $d(u)$.

 \begin{lemma}
\label{Lemma4.1}
 Let $T$ be an ultrametric phylogeny and $u$  an internal node.
 Under the 2-state Jukes-Cantor
 model, for any path $P(x, y)$ from an internal node $x$ to its leaf
 descendant $y$,
\begin{eqnarray}
  \prod _{uv\in P(x, y)} (2p_{uv}-1) = e^{-2\lambda d(x)}.
\end{eqnarray}
 \end{lemma}
 \begin{proof}
It follows from that
   $2p_{uv}-1 =e^{-2\lambda t_{uv}}$ for each edge  $uv$ and that
   $d(x)=\sum_{uv\in P(x, y)}t_{uv}$.
 \end{proof}

Let $T$ be an ultrametric tree that has three or more leaves.
%
For any internal node $w$ with children $w_1$ and $w_2$,
  by Formula~(\ref{recurrence_for_D}) and Lemma~\ref{Lemma4.1}, we have that
 \begin{eqnarray}
\label{Formula_Dw}
   D_w & \geq \frac{1}{2} (2p_{ww_1}-1) D_{w_1} + \frac{1}{2}(2p_{ww_2}-1)
   D_{w_2}
 \end{eqnarray}
 because $C_{w_1}, C_{w_2}\geq 0$.
 By induction, 
 we can show the following
fact from Formula (\ref{Formula_Dw}).

\begin{lemma}
\label{Lemma4.2}
  \begin{eqnarray*}
    D_w \geq  \prod_{(u, v)\in P(w, l)} (2p_{uv}-1) =e^{-2\lambda
     d(w)},
  \end{eqnarray*}
  where $l$ is a leaf below $w$.
\end{lemma}

By Formula (\ref{definition}),
the above lemma implies that the accuracy of reconstructing the root state from
all the leaf states  is
not less than from a single leaf. Such a  fact  was  established  by
Fischer and Thatte in \cite{Fischer_Man_2008}. It can be strengthen as
follows.

\begin{theorem}
\label{Thm4}
   Let $T$ be an ultrametric tree having three
   or more leaves and let $x$ be a child of its root $r$. If $x$ has two  children,
 then
\begin{eqnarray}
     D_r\geq e^{-2\lambda d(r)} [1 +\frac{1}{4}(1- e^{-4\lambda d(x)})]
\end{eqnarray}
\end{theorem}
\begin{proof}
By Lemma~\ref{Lemma4.2},
$D_y \geq e^{-2\lambda d(y)}.$
Since  $C_y \geq 0$,
   by Formula (\ref{recurrence_for_D}), we have that
{\small
\begin{eqnarray}
   D_r &\hspace*{-1em} = &\hspace*{-1em} \frac{1}{2} (2p_{rx}-1) (1+ C_y) D_{x}
+
   \frac{1}{2}(2p_{ry}-1)(1+C_x)D_{y} \nonumber \\
   &\hspace*{-1em} \geq  &\hspace*{-1em} \frac{1}{2}(2p_{rx}-1)D_{x} +
\frac{1}{2} e^{-2\lambda d(r)} (1+C_x) \label{Estimate_Dr}
 \end{eqnarray}
}

Let $u$ and $v$ be the children of $x$.
By Lemma~\ref{Lemma4.2}, $D_u\geq e^{-2\lambda d(u)}$
and $D_v\geq e^{-2\lambda d(v)}$.
Let
$$D_u = e^{-2\lambda d(u)} (1+\Delta(u)),~~ D_v = e^{-2\lambda d(v)} (1+\Delta(v)),$$
where $\Delta(u), \Delta(v)\geq 0$.
We then have
{\small
\begin{eqnarray}
    D_x &\hspace*{-1em}=&\hspace*{-1em} \frac{1}{2}(2p_{xu}-1) (1+C_v)D_u+ \frac{1}{2}(2p_{xv}-1)
(1+C_u)D_v \nonumber \\
 &\hspace*{-1em} =  &\hspace*{-1em}  e^{-2\lambda d(x)}
\{\frac{1}{2} [1+C_v+\Delta(u)+C_v\Delta(u)]
\nonumber\\
& & + \frac{1}{2}[1+C_u + \Delta(v)+ C_u\Delta(v) ]\}  \nonumber \\
 & \hspace*{-1em}\geq  &\hspace*{-1em}  e^{-2\lambda d(x)}
\{1+\frac{1}{2}\left[ C_u+ C_v + \Delta (u)+
\Delta (v) \right]\}. \label{Estimate_Dx}
\end{eqnarray}
}
Combining Formulas~(\ref{Estimate_Dr}) and (\ref{Estimate_Dx}) gives that
{\small
\begin{eqnarray*}
   D_r \geq  e^{-2\lambda d(r)}\{1+ \frac{1}{2}C_x + \frac{1}{4}
[C_u + C_v + \Delta (u)+ \Delta (v)]\}
\end{eqnarray*}
}

By Formula (\ref{recurrence_for_C}),
{\small
\begin{eqnarray*}
   C_x &\hspace*{-1em}= &\hspace*{-1em}\frac{1}{2}[1-C_u-C_v + 3C_uC_v -
(2p_{xu}-1)(2p_{xv}-1)D_uD_v] \\
    &\hspace*{-1em}\geq &\hspace*{-1em}\frac{1}{2}[1-C_u-C_v  -
(2p_{xu}-1)(2p_{xv}-1)D_uD_v],  \\
&\hspace*{-1em}= &\hspace*{-1em}\frac{1}{2}\{1-C_u-C_v  -
e^{-4\lambda d(x)} [1+ \Delta(u)] [1+ \Delta(v)]\}.  \\
\end{eqnarray*}
}
We further have that
{\small
\begin{eqnarray*}
  D_r & \hspace*{-1em}\geq &\hspace*{-1em}\frac {1}{4}e^{-2\lambda d(r)}\\
  & \hspace*{-1em}& \hspace*{-1em}\times \{5 + \Delta(u)
+\Delta(v)
   -e^{-4\lambda d(x)} [1+ \Delta(u)] [1+ \Delta(v)]\}.
\end{eqnarray*}
}
Since $d(x) > d(v)$,
{\small
\begin{eqnarray*}
[1+\Delta (v) ]e^{-4\lambda d(x)} \le [1+ \Delta (v)]e^{-2\lambda d(v)} =D_v
\le 1.
\end{eqnarray*}
}
Therefore, we obtain that
{\small
\begin{eqnarray*}
& &\hspace*{-1em}\Delta(u)
+\Delta(v) -e^{-4\lambda d(x)} [1+ \Delta(u)] [1+ \Delta(v)]  \\
&\hspace*{-1em}\ge & \hspace*{-1em}\Delta (u) - e^{-4\lambda d(x)}
[1+ \Delta(u)] \\
&\hspace*{-1em}\ge & \hspace*{-1em}- e^{-4\lambda d(x)}
\end{eqnarray*}
}
and
{\small
\begin{eqnarray*}
   D_r
\geq  \frac{1}{4}  e^{-2\lambda d(r)} (5 - e^{-4\lambda d(x)} ) =
e^{-2\lambda d(r)} [1 + \frac{1}{4}(1 - e^{-4\lambda d(x)})].
\end{eqnarray*}
}
\end{proof}

It is known that there exists an ultrametric tree in which the root state can be
reconstructed more accurately from the states of a subset of four leaves than from all the leaf states.  Let  $l_1, l_2, l_3$ be three leaves in  $T$.
Assume that the least common ancestor (lca)  $t$ of $l_2$ and $l_3$
is not the root
$r$ and has depth $d(t)$.  If  the lca of $l_1$ and $t$ is the root,
then, the accuracy of reconstructing the root state
from these three leaves is $\frac{1}{2} + \frac{1}{2}  e^{-2\lambda d(r)} [1 +
\frac{1}{4}(1 - e^{-4\lambda d(t)})]$, which is at most
$\frac{1}{2} + \frac{1}{2}  e^{-2\lambda d(r)} [1 +
\frac{1}{4}(1 - e^{-4\lambda d(x)})]$ because $d(x)\geq d(t)$.
If the lca of $l_1$ and $t$ is not $r$, the accuracy is even smaller.
Therefore,  Theorem~\ref{Thm4} implies that
the reconstruction of the root state from all the leaf states is at least as
accurately as from the states of any three leaves.


\newpage

\appendix{
\bf Appendix: Proof of Theorem~\ref{mainThm}}

We first have that
\begin{eqnarray}
(p_X \alpha_X + q_X \beta_X) (1-\alpha_Y -\beta_Y) - (q_X \alpha_X +
p_X \beta_X) (1-\alpha_Y -\beta_Y)
  = (2p_X -1) C_YD_X, \label{Ineq1}\\
(1-\alpha_X -\beta_X)(p_Y \alpha_Y + q_Y \beta_Y) - (1-\alpha_X
-\beta_X)(q_Y \alpha_Y + p_Y \beta_Y)
 = (2p_Y-1) C_XD_Y, \label{Ineq2}
\end{eqnarray}
and
\begin{eqnarray}
  & & (p_X \alpha_X + q_X \beta_X) (p_Y \alpha_Y + q_Y \beta_Y) -(q_X \alpha_X
+ p_X \beta_X) (q_Y \alpha_Y + p_Y \beta_Y) \nonumber \\
  &=& (p_X+p_Y -1) (\alpha _X \alpha_Y - \beta_X \beta_Y) + (b-a) (\beta_X \alpha_Y - \alpha_X \beta_Y). \label{Ineq3}
\end{eqnarray}

Since
$$\alpha_X \alpha_Y - \beta_X \beta_Y = (\alpha_X -\beta_X)\alpha_Y + \beta_X(\alpha_Y - \beta_Y)$$
and
$$\beta_X \alpha_Y - \alpha_X \beta_Y = \alpha_X(\alpha_Y - \beta_Y) -(\alpha_X - \beta_Y)\alpha_Y,$$
combining the equalities (\ref{Ineq1})-(\ref{Ineq3}) given above
leads to
\begin{eqnarray*}
 D_{Z} = (2p_{X}-1) (1-\beta_Y) D_X + (2p_Y-1)(1-\alpha_X) D_Y.
\end{eqnarray*}
By symmetry,
\begin{eqnarray*}
 D_{Z} = (2p_{X}-1) (1-\alpha_Y) D_X + (2p_Y-1)(1-\beta_X) D_Y.
\end{eqnarray*}
Therefore,
\begin{eqnarray}
 D_{Z} & = & \frac{1}{2}(2p_{X}-1) (2-\alpha_Y-\beta _Y) D_X + \frac{1}{2}(2p_Y-1)(2-\alpha _X - \beta_X) D_Y \nonumber \\
  & = & \frac{1}{2}(2p_{X}-1) (1 +  C_Y) D_X + \frac{1}{2}(2p_Y-1)(1+
 C_X) D_Y
\end{eqnarray}
Moreover, we also have that
\begin{eqnarray*}
 & & \alpha_Z + \beta_Z\\
 & = & (p_Xp_Y + q_Xq_Y)(\alpha_X\alpha_Y + \beta_X\beta_Y) + (q_Xp_Y+p_Xq_Y)(\beta_X\alpha_Y +
\alpha_X \beta_Y) \\
  & & + (\alpha_X + \beta_X) (1-\alpha_Y -\beta_Y) + (1-\alpha_X -\alpha_Y)(\alpha_Y + \beta_Y).
\end{eqnarray*}
Since $$\alpha_X\alpha_Y +
\beta_X\beta_Y=\frac{1}{2}((1-C_X)(1-C_Y)+D_XD_Y)$$ and
$$\beta_X\alpha_Y + \alpha_X \beta_Y =\frac{1}{2}((1-C_X)(1-C_Y)-D_XD_Y),$$
 we obtain that
\begin{eqnarray*}
  1- C_Z = \frac{1}{2}\left[1 + C_X +C_Y - 3 C_X C_Y + (2p_X-1)(2p_Y-1)D_XD_Y\right],
\end{eqnarray*}
or equivalently
\begin{eqnarray}
 C_Z = \frac{1}{2}\left[1 - C_X - C_Y + 3 C_X C_Y - (2p_X-1)(2p_Y-1)D_XD_Y\right].
\end{eqnarray}

\end{document}